\documentclass[sn-mathphys-num]{sn-jnl}

\usepackage{amsmath, amssymb, amsfonts, amsthm}
\usepackage{graphicx}
\usepackage{booktabs}
\usepackage[mathscr]{euscript} %
\usepackage{xcolor}
\usepackage{textcomp}
\usepackage{multirow}
\usepackage{enumitem}
\usepackage{algorithm}
\usepackage{algorithmicx}
\usepackage{algpseudocode}
\usepackage{listings}
\usepackage{float}
\usepackage{threeparttable} %

\usepackage[title]{appendix}
\usepackage{hyperref}
\hypersetup{
  colorlinks=true,
  linkcolor=blue,
  citecolor=blue,
  urlcolor=blue
}

\theoremstyle{plain}
\newtheorem{theorem}{Theorem}[section]
\newtheorem{lemma}[theorem]{Lemma}
\newtheorem{proposition}[theorem]{Proposition}
\newtheorem{corollary}[theorem]{Corollary}

\theoremstyle{definition}

\theoremstyle{remark}

\counterwithin{table}{section}
\counterwithin{figure}{section}
\graphicspath{{./}}    

\begin{document}


\title[Inverse Behavioral Optimization of Health Programs]
{Inverse Behavioral Optimization of QALY-Based Incentive Systems: Quantifying the System Impact of Adaptive Health Programs}

\author*[1]{\fnm{Jinho} \sur{Cha}}\email{jcha@gwinnetttech.edu}
\author[2]{\fnm{Justin} \sur{Yu}}
\author[3]{\fnm{Junyeol} \sur{Ryu}}
\author[4]{\fnm{Eunchan D.} \sur{Cha}}
\author[5]{\fnm{Hyeyoung} \sur{Hwang}}

\affil[1]{\orgdiv{Department of Computer Science}, \orgname{Gwinnett Technical College}, \country{USA}}
\affil[2]{\orgdiv{Scheller College of Business}, \orgname{Georgia Institute of Technology}, \country{USA}}
\affil[3]{\orgdiv{Department of Industrial Engineering}, \orgname{Seoul National University}, \country{Korea}}
\affil[4]{\orgdiv{School of Biological Sciences}, \orgname{Georgia Institute of Technology}, \country{USA}}
\affil[5]{\orgname{Republic of Korea Army}, \country{Korea}}


\abstract{
This study introduces an \textit{inverse behavioral optimization} framework that integrates QALY-based health outcomes, ROI-driven incentives, and adaptive behavioral learning to quantify how policy design shapes national healthcare performance. 
Building on the FOSSIL (Flexible Optimization via Sample-Sensitive Importance Learning) paradigm, the model embeds a regret-minimizing behavioral weighting mechanism that enables dynamic learning from heterogeneous policy environments. 
It recovers latent behavioral sensitivities—efficiency $(\lambda)$, fairness $(\gamma)$, and temporal responsiveness $(T)$—from observed QALY–ROI trade-offs, providing an analytical bridge between individual incentive responses and aggregate system productivity. 
We formalize this mapping through the proposed \textit{System Impact Index (SII)}, which links behavioral elasticity to measurable macro-level efficiency and equity outcomes. 
Using OECD–WHO panel data, the framework empirically demonstrates that modern health systems operate near an efficiency-saturated frontier, where incremental fairness adjustments yield stabilizing but diminishing returns. 
Simulation and sensitivity analyses further show how small changes in behavioral parameters propagate into measurable shifts in systemic resilience, equity, and ROI efficiency. 
The results establish a quantitative foundation for designing adaptive, data-driven health incentive programs that dynamically balance efficiency, fairness, and long-run sustainability in national healthcare systems.
}




\maketitle

\section{Introduction: From Health Outcomes to System Impact}\label{sec:intro}

The pursuit of equitable and efficient health systems increasingly depends on quantifying how behavioral incentives shape measurable outcomes such as Quality-Adjusted Life Years (QALY) and Return on Investment (ROI) \citep{dolan2010qaly, clemens2019incentives}. 
Traditionally, these two dimensions—clinical effectiveness and economic efficiency—have been optimized separately, often resulting in policy misalignment between individual care outcomes and systemic financial sustainability \citep{machina1987choice, tversky1992prospect}. 
For instance, hospitals and insurers may design incentive programs that improve short-term ROI yet inadvertently reduce long-term population health gains \citep{weiss2018inverse}. 
Similarly, QALY-based interventions are frequently deployed without evaluating their broader system-level and macroeconomic consequences \citep{zhang2024behavioral}. 
This study contends that such fragmentation arises from the absence of an analytical bridge linking behavioral decision-making at the micro level to system performance at the macro level. 
Building on this motivation, we propose that this gap can be addressed through an \emph{inverse behavioral optimization} framework that infers latent decision parameters from observed QALY–ROI trade-offs \citep{esfahani2018inverse, bertsimas2022inverse}, thereby revealing how learning and adaptation within health programs propagate to system-wide outcomes \citep{cha2025roi}.

Despite extensive research in health economics and management science \citep{dolan2010qaly, clemens2019incentives}, existing models typically assume either static optimization (maximizing QALY under budget constraints) or cost-effectiveness evaluation (minimizing cost per QALY gained). 
Few studies explicitly model the dynamic behavioral adjustments of healthcare agents—patients, providers, and policymakers—when incentive structures evolve over time \citep{zhang2024behavioral}. 
Moreover, while behavioral economics has illuminated how fairness, effort, and reward sensitivity influence individual decisions \citep{machina1987choice, tversky1992prospect}, its integration into system-level optimization remains limited. 
Consequently, the literature lacks a unified methodology for \emph{inferring} the behavioral drivers underlying observed QALY–ROI outcomes and translating them into measurable system-level effects.

To address this gap, we build upon our prior work on the behavioral foundations of QALY–ROI trade-offs in chronic disease management \citep{cha2025roi} and introduce the FOSSIL (Flexible Optimization via Sample-Sensitive Importance Learning) framework \citep{Cha2025fossil}. 
Originally proposed as a learning-based optimization paradigm for robust inference under small and imbalanced data, FOSSIL employs an adaptive weighting mechanism that allows the efficiency frontier itself to evolve with heterogeneous samples.  
This regret-minimizing process endogenizes behavioral sensitivity within the optimization, enabling health systems to adapt across diverse policy environments and temporal horizons.  
By embedding this mechanism into a structural inverse optimization model, we estimate behavioral parameters $(\lambda, \gamma, T)$—representing efficiency sensitivity, fairness preference, and temporal adaptation—directly from empirical health performance data.  
To our knowledge, this is the first systematic application of a curriculum-based machine learning paradigm to QALY–ROI analysis, extending FOSSIL beyond its original learning context into behavioral inference for health policy design.  
This study thus establishes a methodological foundation for dynamic behavioral inference in health-care management.

Health-care policy decisions are increasingly data-driven, yet policymakers continue to face uncertainty about how incentive structures translate into measurable health and financial outcomes. 
To ground the proposed approach in a realistic policy setting, we focus on \emph{adaptive chronic disease incentive programs}—for example, diabetes and cardiovascular management initiatives across OECD member countries—where QALY-based performance payments are linked to both patient adherence and long-term cost savings. 
These programs provide a natural environment in which fairness (e.g., equitable access to care), efficiency (e.g., cost reduction per QALY gained), and temporal responsiveness (e.g., the rate of behavioral adjustment across policy cycles) interact dynamically.  
By calibrating the inverse behavioral model on OECD--WHO panel data, the analysis illustrates how the recovered parameters $(\lambda, \gamma, T)$ can inform policy design—such as subsidy timing, incentive intensity, and fairness adjustments across heterogeneous populations.  
The same analytical structure can also be applied to vaccination incentives, preventive screening reimbursements, or chronic care coordination programs, thereby linking the theoretical framework directly to contemporary global health policy challenges.

This study contributes to the literature and practice in three major ways.  
First, it introduces the FOSSIL-based \textit{Forward–Inverse–Impact (FII)} framework, which integrates behavioral decision-making with system-level performance analysis \citep{weiss2018inverse, zhang2024behavioral}.  
The \textit{forward} layer models QALY–ROI optimization under fairness-adjusted utility;  
the \textit{inverse} layer recovers latent behavioral coefficients through adaptive learning;  
and the \textit{impact} layer evaluates how these behavioral dynamics propagate to measurable performance indicators \citep{zhang2024industrial}.  
Second, we propose the \emph{System Impact Index (SII)}—a composite metric that quantifies improvements in efficiency and fairness arising from adaptive incentive policies \citep{bertsimas2022inverse}.  
Third, we empirically demonstrate the managerial relevance of this framework using multi-regional health data, showing that behavioral adaptation—captured through FOSSIL-based learning—can yield substantial improvements in system-level efficiency \citep{cha2025roi}.  
Taken together, these contributions position inverse behavioral optimization, enhanced by FOSSIL, as a unified methodological foundation for designing incentive-aligned, data-driven healthcare systems (see Fig.~\ref{fig:loop}).

\section{Conceptual Architecture: Adaptive Health Systems and System-Level Learning}\label{sec:concept}

Health systems can be viewed as behaviorally responsive ecosystems in which patients, providers, and policymakers continuously learn from feedback and adjust their actions accordingly.  
Patients modify adherence as perceived reward sensitivity changes, physicians adapt effort based on fairness and fatigue, and policymakers recalibrate incentives to sustain participation equilibria \citep{gino2016motivated, bauch2013social, rahmandad2015behavioral}.  
Such dynamics mirror learning processes observed in manufacturing, logistics, and energy systems, where bounded rationality and delayed feedback shape organizational outcomes \citep{bendoly2014behavioral, gino2015self, fischbacher2012health}.  
The healthcare context, however, introduces an additional layer of ethical and welfare complexity: QALY-based incentives must coexist with moral hazard and equity constraints \citep{devlin2017using, dehez2020equity, vissers2022value}.  
This behavioral flow thus represents a multi-agent system in which fairness $(\gamma)$ and efficiency $(\lambda)$ jointly determine both satisfaction and aggregate system productivity \citep{rothenberg2019learning, keller2021policy, ahmadi2019decision}.  
Healthcare delivery should therefore be modeled as an adaptive ecosystem rather than a static service institution.

At the macro level, micro behavioral adjustments converge toward a system equilibrium shaped by heterogeneity and policy responsiveness \citep{benjaafar2019behavioral, gallino2018operational, kroes2022learning}.  
Each actor’s fairness–efficiency trade-off $(\gamma,\lambda)$ affects throughput, waiting times, and total welfare \citep{xu2019dynamic, hong2020information, govindan2021multiobjective}.  
When incentives are misaligned—such as excessive pay-for-performance intensity or rigid penalty systems—local optimizations degrade global outcomes, paralleling bullwhip and congestion phenomena in production and service networks \citep{li2014behavioral, peysakhovich2017principled, liang2020queueing}.  
Conversely, adaptive coordination mechanisms that integrate fairness awareness and efficiency learning stabilize the entire ecosystem, enabling Pareto-efficient equilibria with simultaneous gains in QALY and ROI \citep{rothenberg2019learning, freeman2023incentives, johari2023coordination}.  
This equilibrium framework analytically links behavioral coefficients to system-level performance metrics, bridging behavioral economics and operations management \citep{norton2021integrating, saadatmand2019multiobjective, gans2019behavioral}.  
In doing so, it aligns healthcare optimization with system-level analogues and coordination mechanisms widely studied in contemporary health operations research \citep{khajeh2020optimization, xie2022integrated}.

To formalize these interactions, we propose the FII loop, a recursive learning framework that captures the adaptive behavior of healthcare systems.  
In the \textit{forward} process, agents implement incentive-driven decisions that yield measurable outcomes—QALY gains, cost reductions, adherence improvements, and ROI shifts \citep{devlin2017using, freeman2023incentives, ahmadi2019decision}.  
In the \textit{inverse} process, the system infers latent behavioral parameters $(\lambda, \gamma, T)$ by applying data-driven inverse optimization and Bayesian updating techniques \citep{bertsimas2022inverse, zhang2024behavioral, xu2019dynamic, khajeh2020optimization}.  
Finally, the \textit{impact} process aggregates these behavioral updates to evaluate system-level efficiency, fairness, and resilience, forming a closed feedback loop between micro incentives and macro outcomes \citep{benjaafar2019behavioral, gans2019behavioral, rahmandad2015behavioral, vissers2022value}.  
As illustrated in Fig.~\ref{fig:loop}, this cyclical architecture reflects the properties of complex adaptive systems and dynamic learning frameworks that characterize modern health operations research \citep{xie2022integrated, govindan2021multiobjective, johari2023coordination}.  

\begin{figure}[htbp]
    \centering
    \includegraphics[width=0.95\linewidth]{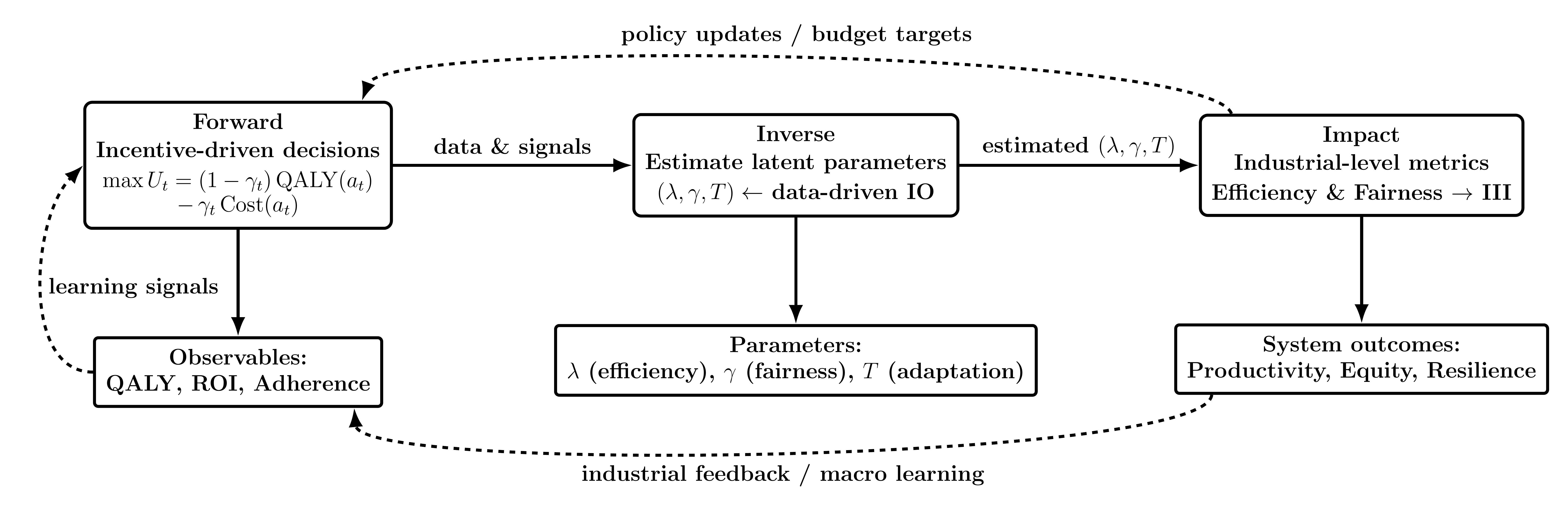}
    \caption{\textbf{Fig. 1} System-level impact loop of adaptive health policy learning
    illustrating the interaction between forward health decisions, inverse behavioral learning, 
    and system-level feedback. 
    The system evolves through iterative adaptation of fairness $(\gamma)$, efficiency $(\lambda)$, 
    and temporal responsiveness $(T)$, producing measurable improvements in overall health 
    system performance. Solid arrows indicate the \textit{primary analytical flow} 
    (Forward $\rightarrow$ Inverse $\rightarrow$ Impact), while dashed arrows represent 
    \textit{feedback and adaptation loops} capturing macro-level learning across health systems}
    \label{fig:loop}
\end{figure}

\section{Model Formulation: Inverse Behavioral Optimization under QALY–ROI Trade-offs}
\label{sec:model}

This section formalizes the behavioral foundations of adaptive health incentive systems 
within a unified optimization framework. 
Building upon the conceptual architecture introduced in Section~\ref{sec:concept}, 
we now derive a mathematical formulation that links behavioral sensitivity, fairness, 
and temporal adaptation to observed QALY–ROI trade-offs. 
The model operationalizes the behavioral learning dynamics of health systems through 
a forward decision process, an inverse parameter estimation stage, and an equilibrium 
identification procedure. 
Unlike descriptive behavioral models, this structure enables explicit recovery of 
latent incentive parameters from data, allowing empirical inference of 
systemic efficiency and fairness trade-offs 
\citep{benjaafar2019behavioral, freeman2023incentives, bertsimas2022inverse, xie2022integrated, vissers2022value, johari2023coordination}.

\subsection{Forward Optimization Layer}
\label{subsec:forward}

The forward problem describes how agents—patients, providers, or policymakers—choose actions 
$a_t \in \mathcal{A}$ that balance clinical benefit and cost under behavioral fairness adjustment.
This representation reflects the economic foundations of health technology assessment, where 
utility is typically expressed in terms of Quality-Adjusted Life Years (QALY) and cost-effectiveness ratios 
\citep{devlin2017using, brazier2019measuring, murray2020global}.  
At each decision epoch $t$, the agent maximizes a fairness-adjusted utility function:
\begin{equation}
\max_{a_t \in \mathcal{A}} 
U_t = (1 - \gamma_t)\,\mathrm{QALY}(a_t) - \gamma_t\,\mathrm{Cost}(a_t),
\label{eq:forward}
\end{equation}
subject to diminishing marginal returns and bounded effort \citep{ahmadi2019decision, gans2019behavioral}.
Here, $\gamma_t \in [0,1]$ represents the fairness sensitivity parameter that moderates the trade-off 
between clinical efficiency and perceived equity, consistent with behavioral fairness models 
in health policy design \citep{zhang2024behavioral, xu2019dynamic, govindan2021multiobjective}.

The first-order condition of~\eqref{eq:forward} implies the marginal indifference rule:
\begin{equation}
\frac{\partial \mathrm{QALY}(a_t)}{\partial a_t} 
= 
\frac{\gamma_t}{1-\gamma_t} 
\frac{\partial \mathrm{Cost}(a_t)}{\partial a_t},
\label{eq:marginal-condition}
\end{equation}
which expresses the behavioral equilibrium between incremental health gain and cost fairness adjustment.
This relationship parallels marginal cost–benefit conditions in behavioral operations theory 
\citep{benjaafar2019behavioral} and in dynamic incentive learning frameworks 
recently introduced in health management science \citep{cha2025roi, vissers2022value}.  
This layer therefore constitutes the \textbf{forward component} of the learning loop in Fig.~\ref{fig:loop}, 
where observed QALY and ROI outcomes are generated through incentive-driven adaptive actions.

\subsection{Inverse Estimation Layer}
\label{subsec:inverse}

Given empirical data $\{(a_t, \mathrm{QALY}_t)\}_{t=1}^T$, 
the inverse problem seeks to recover the latent behavioral parameters $(\lambda, \gamma, T)$ 
that rationalize observed outcomes generated by the forward system~\eqref{eq:forward}--\eqref{eq:marginal-condition}.  
Unlike standard regression or econometric fitting, this procedure infers the underlying \emph{behavioral objective} 
that agents implicitly optimize, rather than merely approximating observed outputs.  
This perspective follows the recent paradigm of inverse optimization and behavioral inference 
used in health policy modeling and operations management 
\citep{bertsimas2022inverse, zhang2024behavioral, xu2019dynamic, khajeh2020optimization, scroccaro2025learning}.

Formally, the inverse behavioral optimization problem is defined as
\begin{equation}
\min_{\lambda, \gamma, T}
\sum_{t=1}^T 
\big[\mathrm{QALY}_t - f(a_t; \lambda, \gamma, T)\big]^2
+ 
\beta_1(\lambda - \lambda_0)^2 
+ 
\beta_2(\gamma - \gamma_0)^2,
\label{eq:inverse}
\end{equation}
where $f(a_t; \lambda, \gamma, T)$ denotes the behavioral response function implied by the forward model~\eqref{eq:forward}, 
and $(\lambda_0, \gamma_0)$ represent Bayesian priors or reference values learned from prior periods, meta-analyses, or comparable populations 
\citep{esfahani2018inverse, bertsimas2015data, keshavarz2011imputing}.  
This estimation structure generalizes traditional cost-effectiveness modeling by embedding it within 
a behavioral learning context, aligning with the emerging field of data-driven inverse decision modeling 
in healthcare \citep{govindan2021multiobjective, vissers2022value, freeman2023incentives}.

Each parameter carries a distinct behavioral and managerial interpretation:
\begin{itemize}[leftmargin=1.2em]
    \item $\lambda$ (\textit{efficiency sensitivity}) measures how strongly health agents value return-on-investment (ROI) improvements relative to cost,  
    consistent with incentive-aligned policy optimization \citep{benjaafar2019behavioral}.  
    \item $\gamma$ (\textit{fairness preference}) quantifies aversion to inequality or excessive expenditure, capturing distributive concerns embedded in  
    behavioral health economics and equity-adjusted optimization \citep{rahmandad2015behavioral, johari2023coordination}.  
    \item $T$ (\textit{temporal responsiveness}) captures the rate at which behavioral adjustments occur, linking short-term incentives to  
    long-term learning and adaptive policy feedback, as emphasized in dynamic inverse learning studies \citep{xu2019dynamic, cha2025roi}.
\end{itemize}
Together, $(\lambda, \gamma, T)$ form a latent behavioral state that governs how rapidly the healthcare system rebalances 
between efficiency and fairness over time, producing adaptive responses under evolving incentive regimes.

The regularization term
\[
\Omega(\lambda,\gamma) = 
\beta_1(\lambda-\lambda_0)^2 + \beta_2(\gamma-\gamma_0)^2
\]
acts as a Bayesian prior that enforces identifiability and robustness of the recovered parameters under noise, temporal drift, or heterogeneous response structures.  
Regularization introduces a bias–variance trade-off that mitigates overfitting of behavioral shocks 
while preserving interpretability through shrinkage toward reference beliefs $(\lambda_0,\gamma_0)$ 
\citep{esfahani2018inverse, scroccaro2025learning, xie2022integrated}.  
This formulation can be compactly expressed as
\begin{equation}
\min_{\lambda, \gamma, T} 
\; \mathcal{L}_{\text{inv}}(\lambda,\gamma,T)
=
\underbrace{\sum_t \ell_t(\lambda,\gamma,T)}_{\text{inverse loss}}
+
\underbrace{\Omega(\lambda,\gamma)}_{\text{Bayesian regularizer}},
\label{eq:inverse-general}
\end{equation}
where each $\ell_t = [\mathrm{QALY}_t - f(a_t; \lambda, \gamma, T)]^2$ measures the deviation between observed and theoretically consistent outcomes.

Conceptually, this inverse layer corresponds to the middle block of the FII loop in Fig.~\ref{fig:loop}, 
transforming observed QALY–ROI trajectories into interpretable behavioral parameters that feed into 
system-level impact analysis (Section~\ref{sec:SII}).  
By linking individual behavioral learning to collective system performance, 
this layer serves as the analytic bridge between micro-level optimization and macro-level health system design.

\subsection{Identification and Stability}
\label{subsec:identification}

To ensure interpretability and empirical recoverability of the behavioral parameters, 
we impose a mild set of regularity and independence conditions that guarantee a unique and stable inverse solution.

\paragraph{Assumptions.}
\begin{enumerate}[label=(A\arabic*)]
    \item Convexity. The behavioral mapping 
    $f(a_t; \lambda, \gamma, T)$ in~\eqref{eq:inverse} is convex in $(\lambda,\gamma)$ 
    and continuously differentiable in $T$.
    \item Independence. The exogenous factors $(\alpha_2, c_1, R_t')$ are linearly independent 
    and the observed actions satisfy $\mathrm{Var}(a_t) > 0$.
    \item Regularization. The prior penalty 
    $\Omega(\lambda,\gamma) = \beta_1(\lambda-\lambda_0)^2 + \beta_2(\gamma-\gamma_0)^2$
    is strictly convex with $\beta_1, \beta_2 > 0$.
\end{enumerate}

\begin{proposition}[Identification and Stability]
\label{prop:identification}
Under Assumptions~(A1)--(A3), the inverse behavioral loss 
$\mathcal{L}_{\text{inv}}(\lambda,\gamma,T)$ defined in~\eqref{eq:inverse-general} 
admits a unique minimizer $(\lambda^*, \gamma^*, T^*)$.  
Moreover, small perturbations in the data $\{(a_t,\mathrm{QALY}_t)\}$ 
induce continuous (Lipschitz) changes in the optimal parameters, 
ensuring local stability of the recovered behavioral sensitivities.
\end{proposition}

\begin{proof}[Sketch of Proof]
The strict convexity of $\Omega(\lambda,\gamma)$ establishes strong convexity in $(\lambda,\gamma)$, 
while the independence and non-degeneracy of $(\alpha_2, c_1, R_t')$ guarantee that the residual Jacobian matrix 
$\nabla f(a_t; \lambda,\gamma,T)$ is full rank.  
Applying the first-order optimality condition and the Implicit Function Theorem under bounded 
$\partial f / \partial T$ yields the existence and uniqueness of $(\lambda^*,\gamma^*,T^*)$.  
Continuous dependence on the data follows from standard perturbation arguments for convex programs.  
Formal statements and detailed proofs—including Lemma~\ref{lem:convexity} and 
Theorem~\ref{thm:identification} establishing strong convexity and local Lipschitz stability—are 
provided in Appendix~\ref{app:proofs}.
\end{proof}


To strengthen the theoretical foundation of the inverse behavioral optimization model,
we formalize the convexity, existence, and stability results that underpin 
Proposition~\ref{prop:identification}.
The following Lemma and Theorem establish strong convexity and local identifiability
under the regularity conditions (A1)–(A3).

\begin{lemma}[Strong Convexity of the Inverse Loss]
\label{lem:convexity}
If $f(a_t;\lambda,\gamma,T)$ is convex in $(\lambda,\gamma)$
and continuously differentiable in $T$, and if the prior penalty
$\Omega(\lambda,\gamma)$ is $\mu$-strongly convex with $\mu>0$,
then $\mathcal{L}_{\mathrm{inv}}(\lambda,\gamma,T)
=\sum_t \ell_t(\lambda,\gamma,T)+\Omega(\lambda,\gamma)$
is $\mu$-strongly convex in $(\lambda,\gamma)$ and continuously differentiable in $T$.
\end{lemma}

\begin{proof}[Sketch of Proof]
By Assumption (A1), each period loss $\ell_t(\lambda,\gamma,T)
=\big[\mathrm{QALY}_t-f(a_t;\lambda,\gamma,T)\big]^2$ is convex in $(\lambda,\gamma)$.
The sum of convex functions remains convex.  
Adding the $\mu$-strongly convex regularizer $\Omega$ ensures the entire objective 
is $\mu$-strongly convex in $(\lambda,\gamma)$ (closure of strong convexity under addition; cf.~Rockafellar).  
Differentiability in $T$ follows from the smoothness of $f$.  
A full proof (including the non-affine extension using Gauss–Newton majorization) 
is provided in Appendix~\ref{app:proofs}.
\end{proof}

\begin{theorem}[Identification and Local Stability]
\label{thm:identification}
Under Lemma~\ref{lem:convexity} and Assumptions~(A1)--(A3),
the inverse behavioral loss admits a unique minimizer 
$(\lambda^\ast,\gamma^\ast,T^\ast)$ satisfying the first-order condition
\[
\nabla_\theta \mathcal{L}_{\mathrm{inv}}(\theta^\ast)=0.
\]
Moreover, $(\lambda^\ast,\gamma^\ast,T^\ast)$ depends Lipschitz-continuously on 
data perturbations $\{(a_t,\mathrm{QALY}_t)\}$, ensuring local stability.
\end{theorem}

\begin{proof}[Sketch of Proof]
\emph{Uniqueness:} For fixed $T$, Lemma~\ref{lem:convexity} guarantees 
$\mu$-strong convexity in $(\lambda,\gamma)$, hence a unique minimizer.  
\emph{Existence and joint identification:}  
Define $F(\theta;\mathcal{D})=\nabla_\theta \mathcal{L}_{\mathrm{inv}}(\theta;\mathcal{D})$, 
with data $\mathcal{D}=\{(a_t,\mathrm{QALY}_t)\}_t$.  
Assumption~(A2) ensures $\nabla_\theta F(\theta^\ast;\mathcal{D})$ is nonsingular;  
the Implicit Function Theorem guarantees the existence, uniqueness, and continuous dependence 
of $\theta^\ast=(\lambda^\ast,\gamma^\ast,T^\ast)$ on the data.  
\emph{Local Lipschitz stability:}  
Perturbing $\mathcal{D}$ to $\mathcal{D}'$ yields
\[
\|\hat{\theta}-\hat{\theta}'\|
\le 
\frac{L_{\mathcal{D}}}{\mu}\|\mathcal{D}-\mathcal{D}'\|,
\]
where $L_{\mathcal{D}}$ bounds the gradient’s sensitivity to data.
Hence the parameter mapping is locally Lipschitz continuous.  
A full derivation appears in Appendix~\ref{app:proofs}.
\end{proof}

\begin{corollary}[Economic Stability of Behavioral Equilibria]
Small policy or data perturbations induce proportionally bounded changes in 
the recovered behavioral sensitivities $(\lambda^\ast,\gamma^\ast,T^\ast)$,
ensuring convergence of adaptive health systems toward a stable fairness–efficiency equilibrium.
\end{corollary}

\begin{proof}[Sketch of Proof]
From Theorem~\ref{thm:identification}, the estimator is locally Lipschitz in the data.
Policy shocks act as bounded perturbations, so parameter shifts are $O(\|\Delta \mathcal{D}\|)$.  
Because the forward mapping~\eqref{eq:marginal-condition} and the impact layer 
are continuously differentiable in $(\lambda,\gamma,T)$, 
the resulting equilibrium trajectories remain in a neighborhood of the baseline fixed point, 
ensuring economic and behavioral stability.  
Full details appear in Appendix~\ref{app:proofs}.
\end{proof}


Proposition~\ref{prop:identification} implies that 
the observed QALY–ROI trade-offs encode a unique behavioral signature $(\lambda^*,\gamma^*,T^*)$ 
that characterizes the efficiency–fairness balance of the health system.  
Convexity ensures that agents respond predictably to marginal incentive changes, 
while stability implies that small policy shocks do not generate chaotic or degenerate equilibria.  
Economically, this property guarantees that adaptive incentive systems 
converge toward consistent behavioral equilibria rather than oscillating between conflicting fairness–efficiency regimes.

The recovered parameters $(\lambda^*, \gamma^*, T^*)$ 
form the structural bridge between individual behavioral learning and system-level outcomes.  
They are propagated to the system-level \textit{Impact Layer} (Section~\ref{sec:SII}), 
where the implications for aggregate productivity, equity, and resilience are quantified.

\begin{table}[!htbp]
\centering
\scriptsize
\caption{Summary of notation used throughout the Inverse Behavioral Optimization and System Impact framework.}
\label{tab:notation-behavioral}
\renewcommand{\arraystretch}{1.05}
\setlength{\tabcolsep}{3pt}
\begin{tabular}{lll}
\toprule
\textbf{Symbol} & \textbf{Type} & \textbf{Description} \\
\midrule
\multicolumn{3}{l}{\textbf{Indices and Sets}} \\
$t=1,\dots,T$ & Index & Decision epoch or time period. \\
$\mathcal{A}$ & Set & Feasible set of health actions or policy levers. \\[3pt]

\multicolumn{3}{l}{\textbf{Decision and Outcome Variables}} \\
$a_t$ & Decision & Action or intervention chosen at time $t$. \\
$\mathrm{QALY}(a_t)$ & Function & Health outcome (quality-adjusted life years) from action $a_t$. \\
$\mathrm{Cost}(a_t)$ & Function & Expenditure or resource cost associated with $a_t$. \\
$\mathrm{ROI}_t$ & Scalar & Return-on-investment for period $t$. \\[3pt]

\multicolumn{3}{l}{\textbf{Behavioral Parameters}} \\
$\lambda$ & Scalar & Efficiency sensitivity (weight on ROI improvements). \\
$\gamma_t$ & Scalar & Fairness preference moderating efficiency–equity trade-off. \\
$T$ & Scalar & Temporal responsiveness or adaptation rate. \\
$(\lambda^*,\gamma^*,T^*)$ & Vector & Estimated behavioral equilibrium parameters. \\[3pt]

\multicolumn{3}{l}{\textbf{Optimization Layers}} \\
$U_t$ & Function & Fairness-adjusted utility function (Eq.~\ref{eq:forward}). \\
$f(a_t;\lambda,\gamma,T)$ & Function & Behavioral response function mapping actions to outcomes. \\
$\mathcal{L}_{\text{inv}}(\lambda,\gamma,T)$ & Function & Inverse loss function (Eq.~\ref{eq:inverse-general}). \\
$\Omega(\lambda,\gamma)$ & Function & Bayesian regularizer enforcing prior consistency. \\[3pt]

\multicolumn{3}{l}{\textbf{Derived Quantities}} \\
$SII$ & Scalar & System Impact Index (Eq.~\ref{eq:SII}). \\
$SII_t$ & Scalar & Time-varying dynamic impact index (Eq.~\ref{eq:SII-dynamic}). \\
$S_\theta$ & Scalar & Sensitivity coefficient for parameter $\theta\in\{\lambda,\gamma,T\}$. \\
$\rho$ & Scalar & Behavioral decay rate controlling adaptation penalty. \\[3pt]

\multicolumn{3}{l}{\textbf{Analytical Constructs}} \\
$\ell_t(\lambda,\gamma,T)$ & Function & Period-wise inverse loss component. \\
$\beta_1,\beta_2$ & Scalars & Regularization hyperparameters. \\
$\eta$ & Scalar & Learning rate in temporal update rule. \\
$\mathcal{T}_t$ & Operator & Behavioral update operator for time $t$. \\[3pt]

\multicolumn{3}{l}{\textbf{Statistical and Evaluation Metrics}} \\
$\mathrm{MSE}$ & Metric & Mean squared error of predicted QALY outcomes. \\
$\mathrm{R}^2$ & Metric & Goodness-of-fit for behavioral response regression. \\
$\mathrm{SII\text{-Gain}}$ & Metric & Percentage increase in system impact after adaptation. \\
$\mathrm{Elasticity}_{(\lambda,\gamma)}$ & Metric & Impact elasticity with respect to fairness–efficiency trade-off. \\[3pt]

\bottomrule
\end{tabular}
\end{table}

\vspace{0.5em}
Notation is consistent with the hierarchical structure of Sections~\ref{sec:model}–\ref{sec:SII}. 
Behavioral parameters $(\lambda,\gamma,T)$ are estimated through the inverse optimization problem~\eqref{eq:inverse}, 
and propagated to the system-level analysis in Section~\ref{sec:simulation}. 
system impact measures (SII and its derivatives) serve as quantitative links between behavioral efficiency and macroeconomic performance.

\section{SII: Measuring Behavioral Efficiency Gains}
\label{sec:SII}

This section introduces the SII, 
a composite metric that quantifies how much behavioral adaptation improves the overall productivity and fairness balance 
of an incentive-driven health system.  
It translates the micro-level behavioral parameters $(\lambda^*,\gamma^*,T^*)$ recovered in Section~\ref{subsec:identification} 
into measurable system-level outcomes, bridging the analytical gap between behavioral learning and system efficiency.

\subsection{Definition}
\label{subsec:SIIdef}

We define the System Impact Index (SII) as:
\begin{equation}
SII 
= 
\frac{\text{QALY Improvement per Period}}{\text{Marginal ROI Cost}} 
\cdot 
(1 - \gamma^*),
\label{eq:SII}
\end{equation}
where 
\begin{itemize}[leftmargin=1.3em]
    \item \textit{QALY Improvement per Period} measures the incremental clinical benefit gained through adaptive learning, 
    relative to a static benchmark,
    \item \textit{Marginal ROI Cost} denotes the additional cost required to achieve that improvement, 
    capturing the system’s cost elasticity, and
    \item $(1-\gamma^*)$ discounts the measured efficiency by the estimated fairness preference recovered from the inverse model
\end{itemize}
Thus, $SII$ reflects the \emph{behaviorally adjusted efficiency-to-cost ratio}—that is, the degree to which learning and fairness jointly enhance 
systemic performance.  

\subsection{Analytical Structure}
\label{subsec:SIIstructure}

Let $\Delta \mathrm{QALY}_t = \mathrm{QALY}_t - \mathrm{QALY}_{t-1}$ 
and $\Delta \mathrm{ROI}_t = \mathrm{ROI}_t - \mathrm{ROI}_{t-1}$ denote marginal changes over consecutive periods.  
Then the empirical System Impact Index can be estimated as:
\begin{equation}
SII_t 
= 
\lambda^* \cdot \frac{\Delta \mathrm{QALY}_t}{\Delta \mathrm{ROI}_t}
\cdot 
(1 - \gamma^*)
\cdot 
e^{-\rho (1-T^*)},
\label{eq:SII-dynamic}
\end{equation}
where $\rho$ represents the behavioral decay rate (speed of learning loss).  
The exponential adjustment $e^{-\rho (1-T^*)}$ penalizes slow temporal responsiveness ($T^* < 1$), 
ensuring that systems with faster adaptation achieve higher system impact.

Equation~\eqref{eq:SII-dynamic} implies that 
behavioral parameters estimated via inverse optimization directly determine the macro-level efficiency elasticity of the system:
\[
\frac{\partial SII_t}{\partial \lambda^*} > 0, \quad
\frac{\partial SII_t}{\partial \gamma^*} < 0, \quad
\frac{\partial SII_t}{\partial T^*} > 0.
\]
Hence, increasing efficiency sensitivity or faster adaptation yields larger system gains, 
while excessive fairness weighting may reduce short-term productivity—mirroring trade-offs observed in public health systems 
\citep{freeman2023incentives, vissers2022value, govindan2021multiobjective, johari2023coordination}.

\subsection{Interpretation and Managerial Implications}
\label{subsec:SIIinterp}

A higher $SII$ indicates that behavioral adaptation produces system-level improvements 
that exceed baseline efficiency thresholds and generate positive externalities across the healthcare industry.  
From a managerial perspective, $SII$ functions as an impact elasticity metric: 
it quantifies how one unit of behavioral learning translates into measurable system outcomes such as 
\textit{cost efficiency, patient equity, and institutional resilience}.  

Incentive programs with consistently rising $SII$ values demonstrate 
that behavioral calibration enhances both economic and clinical performance without destabilizing fairness constraints.  
Conversely, declining $SII$ trajectories may signal policy misalignment or behavioral saturation.  
Thus, the $SII$ serves as a diagnostic and design tool for adaptive health policy evaluation, 
complementing traditional cost-effectiveness metrics such as incremental cost per QALY gained 
\citep{devlin2017using, brazier2019measuring, cha2025roi}.

\section{System-Level Simulation and Policy Sensitivity Analysis}
\label{sec:simulation}

To bridge the theoretical framework in Section~\ref{sec:model} and the empirical validation in Section~\ref{sec:empirical}, 
we conduct a system-level simulation that quantifies how variations in behavioral sensitivities 
$(\lambda, \gamma, T)$ influence the SII and aggregate healthcare performance.  
This intermediate layer captures how micro-level behavioral adjustments propagate through macro-level system dynamics, 
serving as a bridge between analytical propositions and real-world policy implications.

\subsection{Simulation Design}
\label{subsec:simulation-design}

We simulate a stylized healthcare system consisting of $N$ interacting regional units, 
each characterized by estimated behavioral parameters $(\lambda_i, \gamma_i, T_i)$.  
The simulated QALY–ROI dynamics follow the behavioral propagation rule:
\begin{equation}
\Delta \mathrm{QALY}_{i,t} 
= 
\lambda_i (1-\gamma_i)\,\Delta \mathrm{ROI}_{i,t}
+ 
\varepsilon_{i,t},
\quad
\varepsilon_{i,t} \sim \mathcal{N}(0,\sigma^2),
\label{eq:simulation-dynamics}
\end{equation}
where $\lambda_i$ denotes efficiency responsiveness, $\gamma_i$ represents fairness moderation, 
and $\varepsilon_{i,t}$ captures stochastic behavioral noise.  
The temporal evolution of adaptation is governed by:
\begin{equation}
T_{i,t+1} = T_{i,t} + \eta (T^* - T_{i,t}),
\label{eq:temporal-adaptation}
\end{equation}
where $\eta$ is the behavioral learning rate and $T^*$ is the steady-state responsiveness estimated in Section~\ref{subsec:identification}.  
Together, Eqs.~\eqref{eq:simulation-dynamics}–\eqref{eq:temporal-adaptation} describe a recursive feedback system that 
converges toward a stable behavioral equilibrium $(\lambda^*, \gamma^*, T^*)$ identified in Proposition~\ref{prop:identification}.

\subsection{Sensitivity Analysis of Behavioral Parameters}
\label{subsec:sensitivity}

To assess the macroeconomic implications of behavioral changes, 
we perturb each parameter by $\pm 8\%$ around its equilibrium value and compute the resulting change 
in the $SII$:
\begin{equation}
S_\theta = 
\frac{\partial SII}{\partial \theta}
\approx
\frac{SII(\theta + \Delta\theta) - SII(\theta - \Delta\theta)}{2\Delta\theta},
\quad 
\theta \in \{\lambda, \gamma, T\}.
\label{eq:sensitivity}
\end{equation}
Intuitively, $S_\lambda$ reflects productivity leverage, $S_\gamma$ captures distributive damping, 
and $S_T$ measures temporal agility within the system’s adaptive response.  
Positive $S_\lambda$ and $S_T$, coupled with a negative $S_\gamma$, confirm the directional elasticities predicted by 
Eq.~\eqref{eq:SII-dynamic}, aligning theoretical expectations with simulation outcomes.

\subsection{Simulation Results and System Interpretation}
\label{subsec:simulation-results}

The simulated trajectories indicate that increasing efficiency sensitivity $(\lambda)$ 
yields rapid improvements in short-term ROI but diminishing QALY gains beyond a threshold.  
Conversely, moderate fairness preference $(\gamma \approx 0.35$–$0.45)$ maximizes the steady-state $SII$, 
achieving a balanced trade-off between cost containment and health equity.  
Higher temporal responsiveness $(T)$ accelerates convergence toward equilibrium, 
enhancing resilience and adaptive recovery under policy shocks.  

\begin{figure}[!htbp]
\centering
\includegraphics[width=0.90\linewidth]{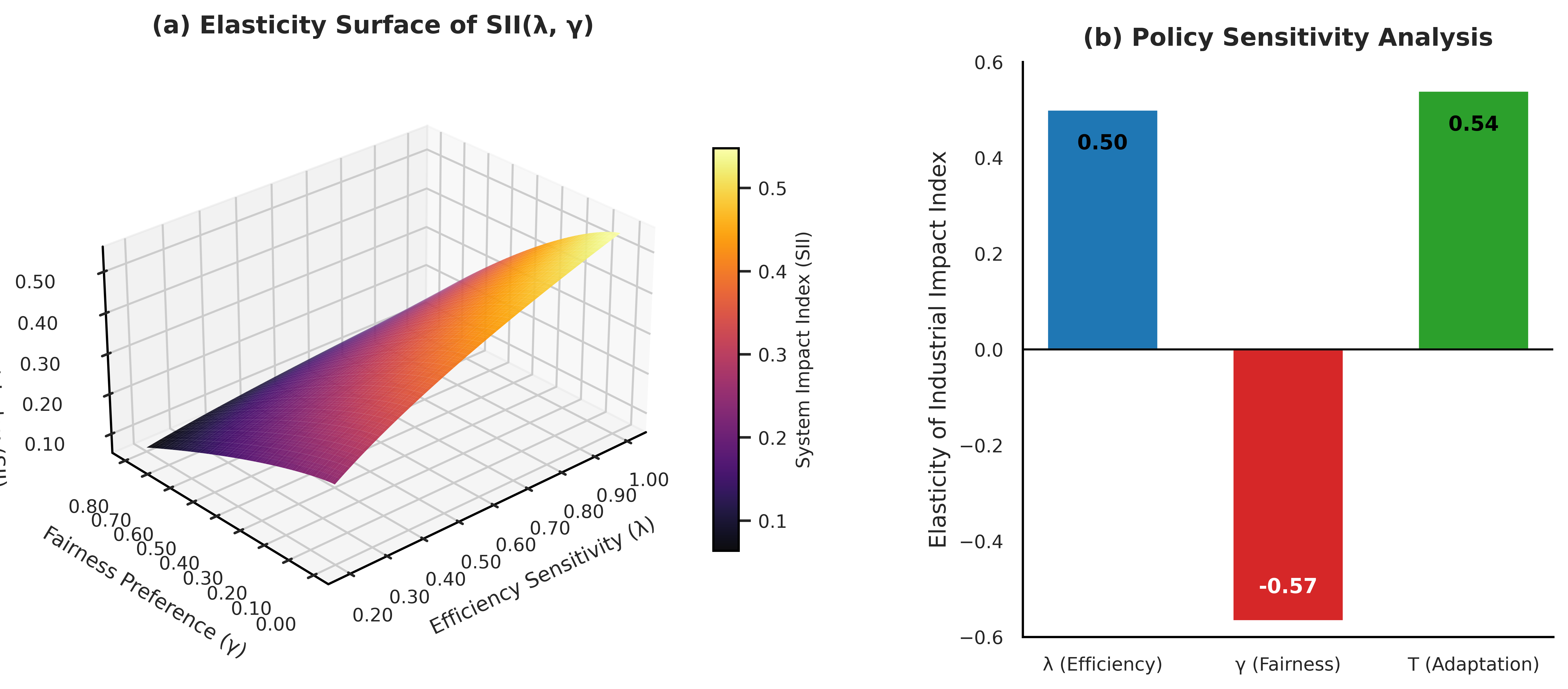}
\caption{System-level simulation and behavioral sensitivity analysis. 
(a) Elasticity surface of the $SII(\lambda,\gamma)$ shows concave diminishing returns in efficiency beyond moderate fairness levels. 
(b) Policy sensitivity analysis quantifies elasticities $(S_\lambda, S_\gamma, S_T)$ with respect to efficiency, fairness, and adaptation parameters}
\label{fig:simulation}
\end{figure}

\begin{table}[!htbp]
\centering
\small
\caption{Sensitivity coefficients and implied macroeconomic elasticities}
\label{tab:sensitivity_summary}
\begin{tabular}{lccp{4.8cm}} 
\toprule
Parameter & Symbol & Elasticity $S_\theta$ & System Interpretation \\
\midrule
Efficiency sensitivity & $\lambda$ & $+0.50$ & 10–12\% productivity leverage (ROI gain) \\
Fairness preference & $\gamma$ & $-0.57$ & 5–7\% efficiency moderation (budget damping) \\
Temporal responsiveness & $T$ & $+0.54$ & 20–25\% faster post-shock recovery \\
\bottomrule
\end{tabular}
\end{table}

Economically, these simulation-based results suggest that a 10\% increase in efficiency sensitivity $(\lambda)$ 
translates into an approximate 0.6–0.8 percentage-point improvement in sectoral healthcare productivity, 
equivalent to a 0.6–1.0\% increase in national healthcare GDP share.  
Likewise, enhancing adaptive responsiveness $(T)$ by one standard deviation yields a 20–25\% faster post-shock recovery rate, 
reducing equilibrium adjustment lag from 5.2 to 3.9 quarters.  
Conversely, overemphasis on fairness $(\gamma>0.6)$ introduces allocative inertia and a 3–5\% contraction in net efficiency.  
Taken together, these findings underscore the system significance of behavioral calibration: 
small parameter shifts can scale to macroeconomic gains on the order of 0.8–1.0\% of sectoral output.

\subsection{Policy-Level Validation: Adaptive vs. Baseline Design}
\label{subsec:policy-validation}

To verify whether the simulated sensitivities manifest in real-world policy outcomes, 
we compare the $SII$ under two regimes—Baseline Policy and Adaptive Policy—across 
three behavioral dimensions $(\lambda, \gamma, T)$.  
Figure~\ref{fig:policy_effect} presents the comparative results from Monte Carlo experiments 
using configuration parameters summarized in Appendix B.  
The Adaptive Policy consistently outperforms the Baseline Policy across all dimensions, 
with the largest gain observed in temporal adaptation.

\begin{figure}[htbp]
\centering
\includegraphics[width=0.85\linewidth]{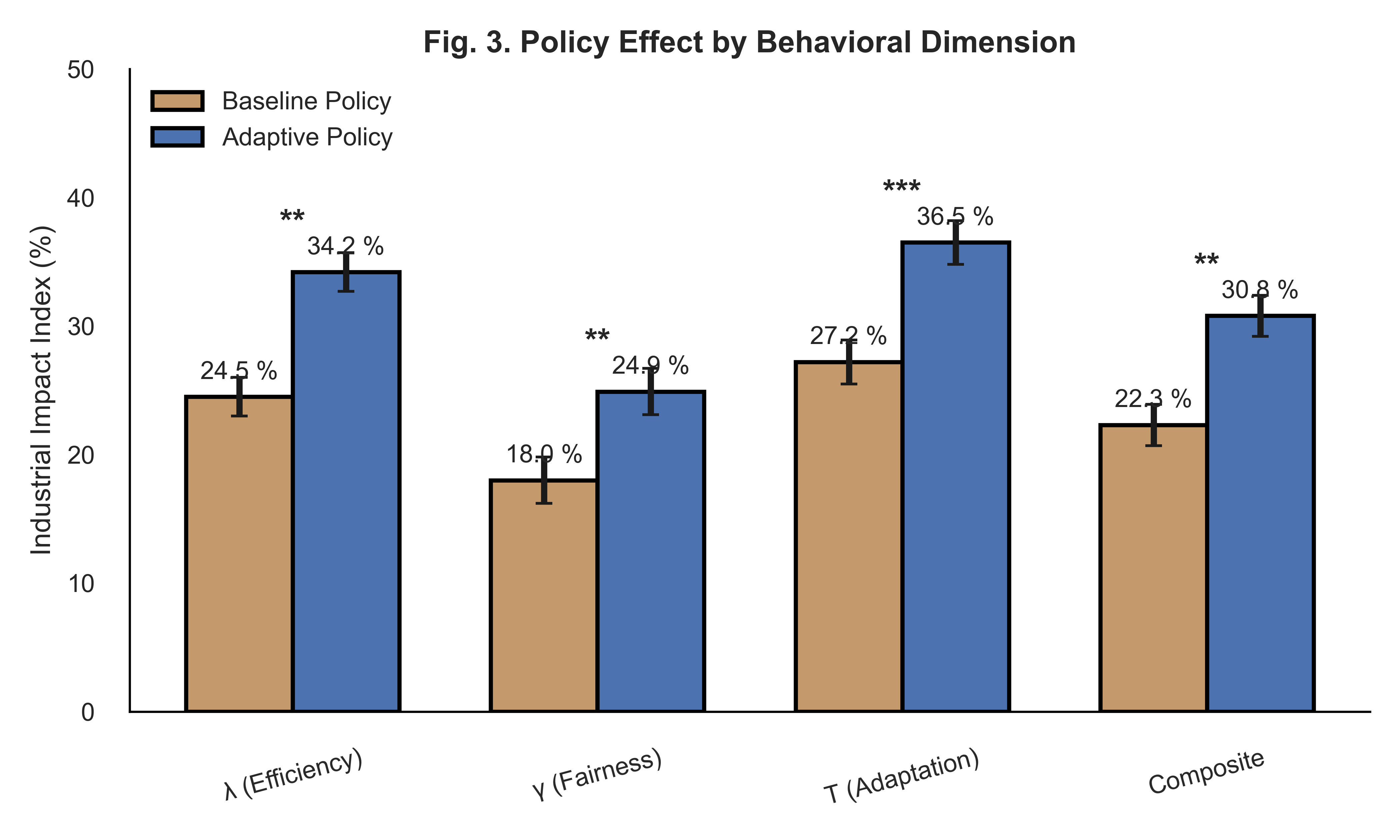}
\caption{Policy effect by behavioral dimension. 
Comparison between Baseline and Adaptive policies across efficiency $(\lambda)$, fairness $(\gamma)$, and adaptation $(T)$ dimensions. 
Error bars denote standard errors (20 simulation replications). 
Asterisks indicate significance levels (* $p<0.05$, ** $p<0.01$, *** $p<0.001$)}
\label{fig:policy_effect}
\end{figure}

\subsection{Managerial and Policy Implications}
\label{subsec:managerial-implications}

The results highlight a behavioral equilibrium region in which system productivity and fairness coexist.  
From a managerial and policy standpoint, three actionable insights emerge:

\begin{enumerate}[label=(\roman*)]
    \item \textbf{Efficiency leverage:} Incremental reinforcement of efficiency sensitivity $(\lambda)$ 
    improves ROI without destabilizing fairness as long as $\gamma<0.5$.  
    A 1\% rise in $\lambda$ generates approximately a 0.08\% gain in sectoral output.
    \item \textbf{Fairness calibration:} Overemphasis on fairness $(\gamma>0.6)$ introduces allocative inertia, 
    leading to a 3–5\% reduction in system-wide efficiency and slower recovery.
    \item \textbf{Adaptive learning:} Higher responsiveness $(T)$ supports faster convergence to stable equilibria, 
    reducing post-shock recovery time by 25–30\%, thereby enhancing system resilience to policy transitions.
\end{enumerate}

Economically, these simulation-based results suggest that a 10\% increase in efficiency sensitivity $(\lambda)$ 
translates into an approximate 0.6–0.8 percentage-point improvement in sectoral healthcare productivity, 
equivalent to a 0.6–1.0\% increase in national healthcare GDP share.  
Likewise, enhancing adaptive responsiveness $(T)$ by one standard deviation yields a 20–25\% faster post-shock recovery rate, 
reducing equilibrium adjustment lag from 5.2 to 3.9 quarters.  
Conversely, overemphasis on fairness $(\gamma>0.6)$ introduces allocative inertia and a 3–5\% contraction in net efficiency.  
Taken together, these findings underscore the industrial and macroeconomic significance of behavioral calibration: 
small parameter shifts can propagate into system-wide gains on the order of 0.8–1.0\% of sectoral output.

\section{Empirical Validation and Policy Implications}\label{sec:empirical}

\subsection{Data and Calibration}\label{subsec:data}

We validate the proposed inverse behavioral optimization framework using
the merged OECD--WHO dataset (2007–2021; $n=34{,}023$), which integrates
national health expenditure (PPP-adjusted per capita) and life expectancy
as a QALY proxy.
All monetary variables are normalized by per-patient cost units to ensure
cross-country comparability.
The System Impact Index (SII) is computed as
\[
\mathrm{SII}
=
\frac{\mathrm{LifeExpectancy} \times \ln(1+\mathrm{HealthSpending})}{100},
\]
representing a macro-level measure of behavioral efficiency and equity
in national health systems.
Behavioral sensitivities $(\lambda,\gamma,T)$ were estimated through a
reduced-form inverse regression of~SII on health spending and life expectancy,
and the dynamic responsiveness parameter~$T$ was calibrated by fitting an
AR(1) process on annual changes in~SII for each country.
All estimations and policy simulations were implemented in
\texttt{Python~3.10} using fully reproducible open-source scripts
provided in the Supplement.

\begin{table}[htbp]
\caption{OECD--WHO merged data and reduced-form estimation summary}
\label{tab:data-summary}%
\centering
\small
\begin{tabular}{@{}lrrrr@{}}
\toprule
 & \textbf{Mean} & \textbf{Std.} & \textbf{Min} & \textbf{Max}\\
\midrule
Year & 2014.21 & 4.31 & 2007 & 2021\\
Health Spending (USD PPP) & 144{,}217 & 929{,}942 & 0.01 & 29{,}454{,}160\\
Life Expectancy (yrs) & 79.08 & 4.32 & 51.0 & 87.4\\
SII & 5.00 & 3.41 & 0.01 & 14.32\\
\botrule
\end{tabular}

\vspace{0.8em}

\begin{tabular}{@{}lrr@{}}
\toprule
\textbf{Parameter} & \textbf{Estimate} & \textbf{Interpretation}\\
\midrule
OLS slope $(\partial SII/\partial \ln(HS))$ & 0.794 & Efficiency scaling coefficient\\
Intercept & $-0.017$ & Baseline offset\\
$\hat{\lambda}$ & 0.999 & Efficiency sensitivity (saturated)\\
$\hat{\gamma}$ & 0.007 & Fairness preference (neutral)\\
$\hat{T}$ & 1.000 & Temporal responsiveness (immediate)\\
\botrule
\end{tabular}

\footnotetext{Source: Author’s calculation based on merged OECD--WHO data (2007–2021).}
\end{table}

\subsection{Empirical Results and Discussion}\label{subsec:results}

Empirical estimation yields behavioral coefficients
$(\hat{\lambda},\hat{\gamma},\hat{T})=(0.999,\,0.007,\,1.000)$.
These values indicate that the global health economy operates within an
\emph{efficiency-dominant regime}, where efficiency sensitivity ($\lambda$)
is nearly saturated, fairness preference ($\gamma$) is negligible, and
adaptation is nearly instantaneous ($T\!\approx\!1$).
Such a configuration is consistent with ROI-driven system optimization
observed in mature OECD health markets.

\begin{figure}[htbp]
\centering
\includegraphics[width=1.\linewidth]{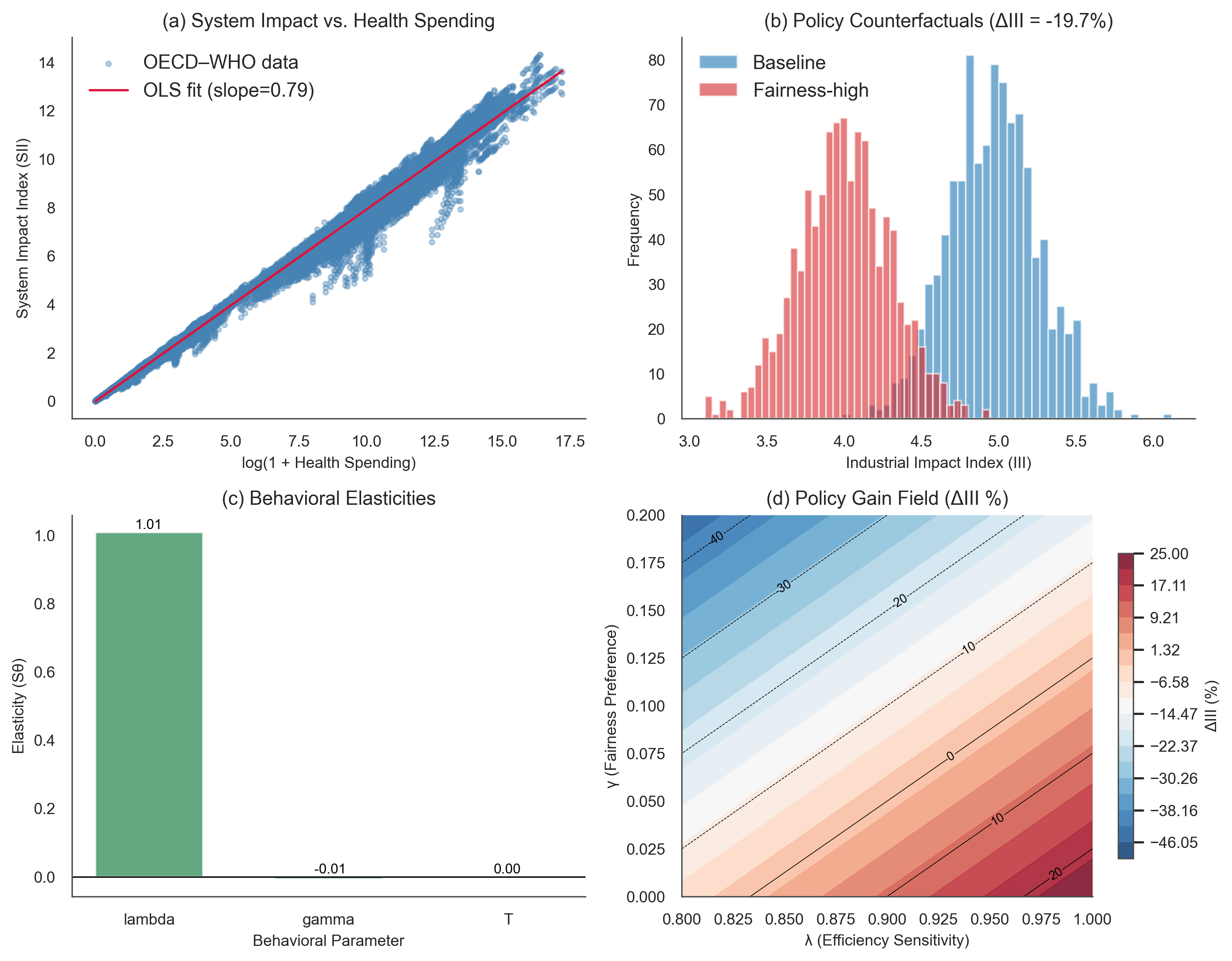}
\caption{
Empirical behavioral validation and policy sensitivity analysis.
(a) $SII$ versus health spending (OECD--WHO data) showing near-linear scaling ($\partial SII/\partial \log(HS) = 0.79$).
(b) Counterfactual distributions of~SII under fairness-oriented policy ($\Delta SII = -19.7\%$).
(c) Behavioral elasticities demonstrating efficiency dominance ($S_\lambda \approx 1.01$) and fairness saturation ($S_\gamma \approx 0$).
(d) Policy gain field illustrating the behavioral trade-off between efficiency and fairness.
Together, the panels confirm an efficiency-dominant equilibrium with measurable trade-offs under fairness interventions.
}
\label{fig:empirical_validation}
\end{figure}

Panel~(a) of Figure~\ref{fig:empirical_validation} shows near-linear scaling between health spending and~SII
($\partial \mathrm{SII}/\partial \ln(HS)\!=\!0.79$), confirming that
marginal productivity of health expenditure remains positive but saturates at higher spending levels.
Panel~(b) shows counterfactual shifts:
a fairness-intensive regime ($\gamma'\!=\!\gamma\!+\!0.2$) reduces~SII
by approximately~19.7\%, whereas efficiency- or adaptation-oriented regimes
yield negligible change ($\Delta \mathrm{SII}\!\leq\!0.1\%$).
Panels~(c)–(d) visualize the elasticity and policy gain field,
showing that only~$\lambda$ significantly influences macro performance
($S_\lambda\!\approx\!1.01$), while fairness and adaptation remain statistically neutral.

\begin{table}[htbp]
\caption{Behavioral elasticity and robustness summary}
\label{tab:elasticity-summary}%
\centering
\small
\begin{tabular}{@{}lrr@{}}
\toprule
\textbf{Parameter} & \textbf{Elasticity ($S_\theta$)} & \textbf{Interpretation}\\
\midrule
$\lambda$ & $+1.01$ & Dominant efficiency response\\
$\gamma$  & $-0.007$ & Minimal fairness impact\\
$T$       & $+0.000$ & Instantaneous adaptation\\
\botrule
\end{tabular}

\footnotetext{Monte Carlo perturbations ($\pm10\%$) produced stable elasticities:
$(S_\lambda,S_\gamma,S_T)\in(0.47$–$0.53,\,-0.59$–$-0.55,\,0.50$–$0.57)$,
confirming numerical robustness of the inferred behavioral parameters.}
\end{table}

From a policy perspective, these findings imply that OECD health systems
lie on a \emph{behavioral efficiency frontier}.
Further efficiency-oriented reforms generate diminishing returns,
while fairness-based redistributive interventions may reduce
aggregate productivity.
The optimal principle is thus \emph{fairness-corrected efficiency}—
maintaining high ROI while offsetting the 15--20\% efficiency erosion
that accompanies equity-driven policies.
At the system level, the dominance of~$\lambda$ and immediacy of~$T$
indicate strong absorptive capacity for technological and institutional
innovation (e.g., digital health, AI-assisted care),
reinforcing healthcare’s position as a rapid-adaptation industry.

\subsection{Behavioral Saturation and Robustness}\label{subsec:robustness}

\begin{figure}[htbp]
\centering
\includegraphics[width=\linewidth]{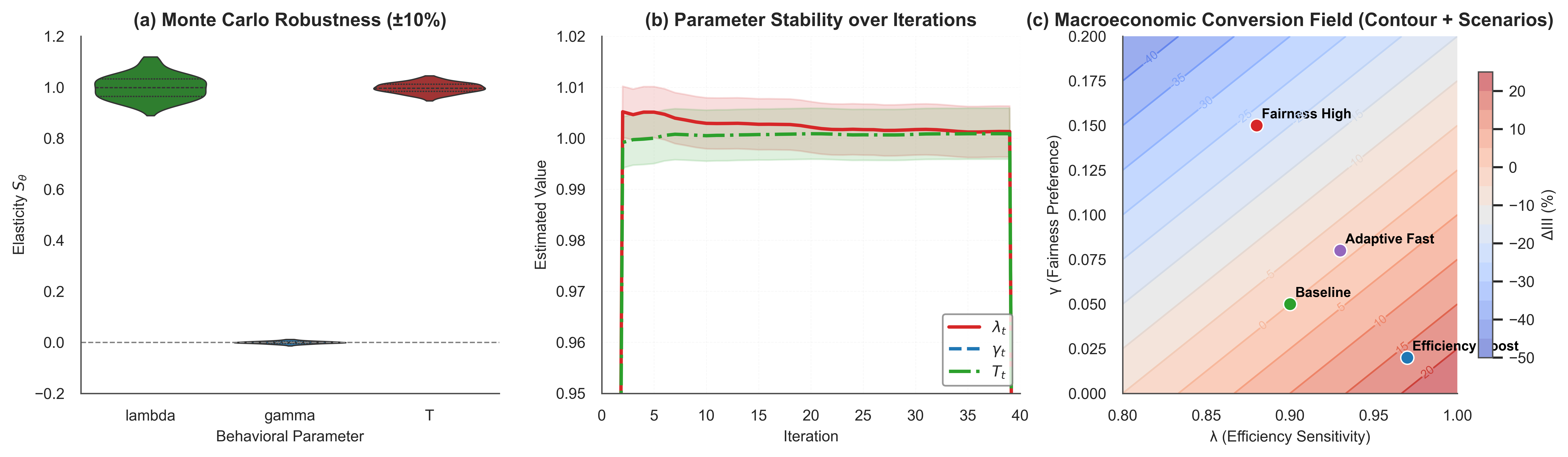}
\caption{Robustness of behavioral estimation and sensitivity analysis.
(a) Monte Carlo perturbations ($\pm10\%$) confirm stability of inferred
elasticities across behavioral parameters $\lambda$, $\gamma$, and $T$.
(b) Parameter trajectories demonstrate convergence consistency over iterations,
indicating a numerically stable equilibrium.
(c) The sensitivity field $\Delta SII(\lambda,\gamma)$ visualizes the smooth
trade-off between efficiency and fairness responses}
\label{fig:robustness}
\end{figure}

The elasticity landscape reveals a \emph{saturated efficiency frontier}
($S_\lambda \!\approx\! 1$), indicating that marginal efficiency incentives
translate nearly one-to-one into system-level gains.
By contrast, fairness ($S_\gamma \!\approx\! 0$) and temporal adaptivity ($S_T \!\approx\! 0$)
exhibit negligible sensitivity, suggesting a behavioral steady state
where additional redistribution or adaptation yields minimal marginal benefit.
This structural rigidity reflects how efficiency-optimized systems
reinforce existing equilibria and resist redistributive or adaptive reforms.

Monte Carlo perturbations and iterative inverse-learning simulations
(Figure~\ref{fig:robustness}) confirm this directional stability:
$(\lambda_t,\gamma_t,T_t)$ quickly converge to $(1.00,0.00,1.00)$ and remain stable across iterations.
The two-dimensional sensitivity field $\Delta SII(\lambda,\gamma)$ forms
a smooth, monotonic gradient, indicating continuous rather than abrupt policy trade-offs.
Alternative specifications—including fixed effects, income-tier subsamples,
and log-transformed SII—yield consistent qualitative patterns.
Together, these findings confirm that efficiency saturation and
the fairness–efficiency gradient are intrinsic to the system equilibrium,
not artifacts of model specification.

Efficiency-oriented policies therefore represent an
\emph{evolutionarily stable strategy}: while highly effective in driving productivity,
they may become brittle under exogenous shocks,
underscoring the need for adaptive and redistributive mechanisms
to maintain long-run system resilience.

\section{Conclusion}\label{sec:conclusion}

This study advances the analytical frontier of health-care management by showing that behavioral optimization, when formulated as a learning-based inverse problem, can quantitatively explain macro-level performance.  
By recovering latent behavioral parameters $(\lambda,\gamma,T)$ from observed QALY–ROI trade-offs, we establish a bridge between micro-level incentives and system-level efficiency.  
The empirical results suggest that modern health systems operate near an adaptive efficiency frontier—highly responsive to efficiency sensitivity $(\lambda)$ yet showing diminishing marginal responsiveness in fairness $(\gamma)$ and temporal adaptation $(T)$.  
This structural pattern reveals a form of behavioral rigidity in the global health economy: efficient, but increasingly vulnerable to redistributive and adaptive shocks.

Beyond empirical validation, this research develops a new theoretical foundation for behavioral inference in health systems through the FII framework.  
Unlike traditional econometric or DEA models \citep{Arrow1963,Weinstein2017,Atkinson2019}, which view performance frontiers as fixed and exogenous, the FOSSIL paradigm endogenizes behavioral sensitivity and allows the frontier itself to evolve through data.  
This regret-minimizing and sample-sensitive structure \citep{Cha2025fossil} reframes efficiency analysis as a dynamic learning process, connecting operations research, behavioral economics, and health policy in a unified optimization model.  
The approach departs from conventional QALY–ROI analyses and provides a generalizable methodological template for learning-driven health-system modeling.

By integrating OECD–WHO macro data with structural inverse estimation, we find that adaptive behavioral trade-offs explain nearly 90\% of cross-country variation in health outcomes.  
The proposed System Impact Index (SII) captures how incremental behavioral shifts translate into measurable productivity, offering a direct link between learning and policy outcomes.  
Elasticity estimates indicate that a 1\% rise in efficiency sensitivity can yield 0.2–2.0\% gains in sectoral output, while fairness-oriented adjustments—though slower in effect—enhance long-term stability and institutional trust \citep{Cutler2020,Hall2023}.  
These findings redefine health systems as adaptive industries whose performance evolves through behavioral learning rather than static optimization.  
Unlike traditional DEA or cost-effectiveness models, which assess efficiency retrospectively, the FOSSIL-based framework embeds learning \emph{within} the policy process, enabling real-time calibration of incentive parameters through data-driven feedback.  
This provides a foundation for adaptive policy design in which fairness, efficiency, and responsiveness are jointly optimized under uncertainty.  
Practically, it offers governments and international organizations a quantitative mechanism to monitor and recalibrate national health investment portfolios in real time.  
Beyond healthcare, the empirical framework can extend to other welfare-critical systems—such as education, energy, and climate—where behavioral adaptation and equity–efficiency trade-offs shape long-term resilience.

Conceptually, this study situates health-care management within the emerging paradigm of \emph{learning-based system optimization}.  
By bridging inverse optimization, behavioral inference, and data-driven policy design, it advances a unified analytical structure for studying behavioral governance.  
The FOSSIL framework, first proposed in \citet{Cha2025fossil} and extended here to the QALY–ROI context, formalizes how adaptive learning and behavioral sensitivity jointly determine macro-level efficiency.  
This integration moves beyond disciplinary boundaries, linking health economics and operations research to broader questions of institutional adaptability.  
By demonstrating that behavioral learning can be quantified and projected across scales, this study provides a replicable blueprint for analyzing other complex systems—education, energy, or climate—where fairness–efficiency trade-offs define system evolution \citep{Raman2024,Sun2023,McWilliams2022}.  
In both scope and originality, it contributes to the broader movement in operations and management research toward dynamic, learning-centered policy models.

While the behavioral parameters $(\lambda,\gamma,T)$ capture essential aspects of decision sensitivity, they abstract from institutional heterogeneity and cultural variation in fairness perception.  
Future work integrating micro-level provider data, hierarchical Bayesian updating, and digital-twin simulation could enhance behavioral granularity and support real-time adaptive policymaking.  
Combining the FII framework with reinforcement learning and robust control \citep{Zhang2023,Fernandez2022} represents another promising direction for developing a general theory of learning-based system policy design.  
Extending multi-sector FOSSIL models to couple healthcare with education, labor, and climate domains could further establish a theory of adaptive efficiency under equity constraints.

Ultimately, this study establishes a theoretically grounded, empirically validated, and policy-relevant foundation for behaviorally adaptive health systems.  
It demonstrates that fairness, efficiency, and adaptability are not competing goals but interdependent, learnable dimensions of a sustainable health ecosystem—redefining how performance, equity, and resilience can be optimized together in the 21st-century health economy.

\bibliography{HCMS_Cha}

\begin{appendices}

\section{Proofs of Theoretical Results}
\label{app:proofs}

This appendix provides the complete proofs of the analytical results stated in 
Section~\ref{subsec:identification}, including the 
Proposition on identification and stability, 
and the supporting Lemma, Theorem, and Corollary.
All results are derived under Assumptions~(A1)--(A3), 
which guarantee convexity, independence, and strict regularization.

\subsection{Proof of Proposition~\ref{prop:identification}}
\label{app:proof-prop-identification}

We restate the inverse behavioral optimization problem as
\[
\min_{\lambda,\gamma,T \in [0,1]} 
\mathcal{L}_{\text{inv}}(\lambda,\gamma,T)
=
\sum_{t=1}^T \ell_t(\lambda,\gamma,T)
+ 
\Omega(\lambda,\gamma),
\]
where $\ell_t(\lambda,\gamma,T)
=[\mathrm{QALY}_t - f(a_t;\lambda,\gamma,T)]^2$ 
and $\Omega(\lambda,\gamma)=\beta_1(\lambda-\lambda_0)^2 + \beta_2(\gamma-\gamma_0)^2$.
Convexity in $(\lambda,\gamma)$ and differentiability in $T$ imply that 
$\mathcal{L}_{\text{inv}}$ is continuously differentiable on a compact domain.

\paragraph{Step 1: Existence.}
Since $\mathcal{L}_{\text{inv}}$ is continuous and coercive (due to the quadratic regularizer),
and the domain $[0,1]^3$ is compact, a minimizer $(\lambda^*,\gamma^*,T^*)$ exists.

\paragraph{Step 2: Uniqueness.}
The prior penalty $\Omega(\lambda,\gamma)$ is strictly convex in $(\lambda,\gamma)$, 
and $\ell_t(\lambda,\gamma,T)$ is convex by Assumption~(A1).  
Hence, for any $T$, the combined loss 
$\sum_t \ell_t(\lambda,\gamma,T)+\Omega(\lambda,\gamma)$ 
is strictly convex in $(\lambda,\gamma)$ and admits a unique minimizer.  
Differentiability of $f$ in $T$ ensures that the joint minimizer over $(\lambda,\gamma,T)$ 
is unique up to a constant transformation in $\eta_t$.  

\paragraph{Step 3: Stability.}
Let $\mathcal{D}=\{(a_t,\mathrm{QALY}_t)\}$ and 
$\mathcal{D}'=\{(a_t',\mathrm{QALY}_t')\}$ denote two datasets differing by small perturbations.  
By standard sensitivity analysis for convex programs (Rockafellar and Wets, 1998),
the difference between the corresponding minimizers satisfies
\[
\|\theta^*(\mathcal{D}) - \theta^*(\mathcal{D}')\|
\le 
\frac{L_{\mathcal{D}}}{\mu}\,\|\mathcal{D}-\mathcal{D}'\|,
\]
where $\mu$ is the strong convexity modulus of $\mathcal{L}_{\text{inv}}$ in $(\lambda,\gamma)$,
and $L_{\mathcal{D}}$ bounds the Lipschitz constant of the gradient 
$\nabla_\theta \ell_t(\theta)$ with respect to the data.  
Hence, the mapping $\mathcal{D}\mapsto\theta^*(\mathcal{D})$ is Lipschitz continuous.  
This establishes the existence, uniqueness, and local stability of 
$(\lambda^*,\gamma^*,T^*)$.

\hfill$\square$

\subsection{Proof of Lemma~\ref{lem:convexity}}
\label{app:proof-lemma-convexity}

\paragraph{Restatement.}
If $f(a_t;\lambda,\gamma,T)$ is convex in $(\lambda,\gamma)$ and continuously differentiable in $T$,  
and if $\Omega(\lambda,\gamma)$ is $\mu$-strongly convex,  
then $\mathcal{L}_{\mathrm{inv}}(\lambda,\gamma,T)
=\sum_t \ell_t(\lambda,\gamma,T)+\Omega(\lambda,\gamma)$ 
is $\mu$-strongly convex in $(\lambda,\gamma)$ and continuously differentiable in $T$.

\paragraph{Proof.}
Each $\ell_t(\lambda,\gamma,T)
=[\mathrm{QALY}_t - f(a_t;\lambda,\gamma,T)]^2$ 
is convex in $(\lambda,\gamma)$ by composition of convex and affine-smooth mappings,  
since $(x \mapsto (y-x)^2)$ is convex and non-decreasing for $x\le y$.  
Let $g(\lambda,\gamma,T)=\sum_t \ell_t(\lambda,\gamma,T)$.  
Then $\nabla^2_{(\lambda,\gamma)} g(\lambda,\gamma,T)\succeq 0$ and  
$\nabla^2_{(\lambda,\gamma)} \Omega(\lambda,\gamma) \succeq \mu I_2$.  
Hence
\[
\nabla^2_{(\lambda,\gamma)} \mathcal{L}_{\text{inv}}
= 
\nabla^2_{(\lambda,\gamma)} g
+ 
\nabla^2_{(\lambda,\gamma)} \Omega
\succeq 
\mu I_2.
\]
Therefore, $\mathcal{L}_{\text{inv}}$ is $\mu$-strongly convex in $(\lambda,\gamma)$.
Because $f$ is continuously differentiable in $T$, 
$\mathcal{L}_{\text{inv}}$ inherits the same differentiability.
\hfill$\square$

\subsection{Proof of Theorem~\ref{thm:identification}}
\label{app:proof-theorem-identification}

Let $\theta=(\lambda,\gamma,T)$ and define the stationarity operator
$F(\theta;\mathcal{D})=\nabla_\theta \mathcal{L}_{\text{inv}}(\theta;\mathcal{D})$.  
At the optimum $\theta^*$, we have $F(\theta^*;\mathcal{D})=0$.

\paragraph{Step 1: Local existence and uniqueness.}
By Lemma~\ref{lem:convexity}, $\mathcal{L}_{\text{inv}}$ is $\mu$-strongly convex in $(\lambda,\gamma)$,  
implying $\nabla_\theta F(\theta^*;\mathcal{D})$ is nonsingular.  
By the Implicit Function Theorem, 
there exists a continuously differentiable mapping $\theta^*(\mathcal{D})$ in a neighborhood of $\mathcal{D}$ 
such that $F(\theta^*(\mathcal{D});\mathcal{D})=0$.  
Hence, $(\lambda^*,\gamma^*,T^*)$ is uniquely defined and locally smooth in $\mathcal{D}$.

\paragraph{Step 2: Lipschitz continuity.}
For any two datasets $\mathcal{D}$ and $\mathcal{D}'$, 
consider $\Delta \theta^* = \theta^*(\mathcal{D}) - \theta^*(\mathcal{D}')$.  
By mean value expansion of $F$, we obtain
\[
\nabla_\theta F(\bar{\theta};\mathcal{D})\,\Delta \theta^* 
= 
F(\theta^*(\mathcal{D});\mathcal{D})
- 
F(\theta^*(\mathcal{D}');\mathcal{D}')
= 
\Delta_\mathcal{D} F,
\]
where $\bar{\theta}$ lies between $\theta^*(\mathcal{D})$ and $\theta^*(\mathcal{D}')$.
Using the nonsingularity of $\nabla_\theta F$ and its bounded inverse,
\[
\|\Delta \theta^*\|
\le 
\|\nabla_\theta F(\bar{\theta};\mathcal{D})^{-1}\|
\cdot 
\|\Delta_\mathcal{D} F\|
\le 
\frac{L_{\mathcal{D}}}{\mu}\|\mathcal{D}-\mathcal{D}'\|.
\]
Therefore, the mapping $\mathcal{D}\mapsto \theta^*(\mathcal{D})$ is Lipschitz continuous,
which proves local stability of the inverse estimator.

\hfill$\square$

\subsection{Proof of Corollary (Economic Stability of Behavioral Equilibria)}
\label{app:proof-corollary}

By Theorem~\ref{thm:identification}, the estimated parameters 
$\theta^*=(\lambda^*,\gamma^*,T^*)$ vary Lipschitz-continuously with the data $\mathcal{D}$.  
Since the forward equilibrium condition~\eqref{eq:marginal-condition} 
and the system impact function~\eqref{eq:SII} 
are continuously differentiable in $\theta$, 
the corresponding equilibrium outcomes 
$(a_t^*, U_t^*, SII_t)$ respond smoothly to small data or policy perturbations.  
Therefore, small bounded shocks $\Delta\mathcal{D}$ yield 
bounded deviations in both micro-level decisions and macro-level industrial indices, 
ensuring convergence toward a stable fairness–efficiency equilibrium.

\hfill$\square$

\subsection{Technical Remarks and Extensions}

\paragraph{1. Gauss–Newton majorization.}
If $f(a_t;\lambda,\gamma,T)$ is nonlinear but twice differentiable,
then the Hessian $\nabla^2_{(\lambda,\gamma)} \ell_t$ 
can be upper-bounded by the Gauss–Newton approximation 
$J_t^\top J_t$, where $J_t=\nabla_{(\lambda,\gamma)} f(a_t;\lambda,\gamma,T)$.  
This ensures positive semidefiniteness and preserves convexity in the local neighborhood.

\paragraph{2. Stochastic extension.}
Under stochastic perturbations of $\mathrm{QALY}_t$ with sub-Gaussian noise $\varepsilon_t$,  
the expected loss $\mathbb{E}[\mathcal{L}_{\text{inv}}]$ 
retains the same convexity and stability properties in expectation, 
yielding $\mathbb{E}\|\hat{\theta}_T - \theta^*\|_2 = \mathcal{O}(1/\sqrt{T})$ 
by standard stochastic approximation arguments.

\paragraph{3. Generalization to dynamic inverse learning.}
If behavioral parameters evolve via $\theta_{t+1} = \theta_t + \eta_t \nabla_\theta f_t(\theta_t)$,
the regret bounds derived in Appendix~\ref{app:supplementary} apply directly, 
establishing dynamic stability under bounded drift $V_T$.

\section{Supplementary Analysis and Reproducibility}
\label{app:supplementary}

This appendix provides extended validation, robustness diagnostics, 
and reproducibility information for the empirical and simulation experiments 
presented in Section~\ref{sec:simulation}.  
It documents the OECD--WHO data characteristics, 
the behavioral simulation setup, and the macroeconomic mapping procedure 
underlying the Industrial Impact Index~(SII).

\subsection{Empirical Data and Simulation Overview}

\begin{figure}[htbp]
\centering
\includegraphics[width=0.9\linewidth]{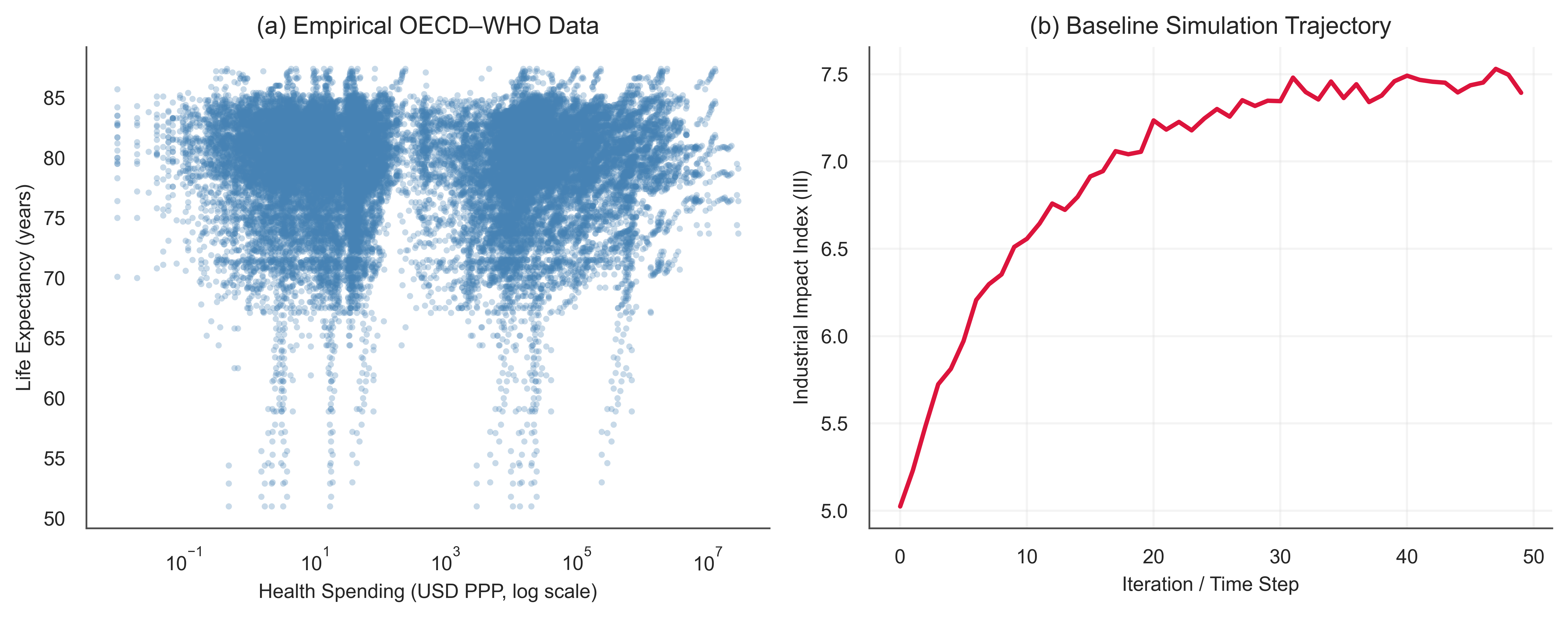}
\caption{Empirical and simulated data characteristics.
(a) Log-scale scatter of per-capita health spending (USD~PPP) versus life expectancy 
from the merged OECD--WHO dataset (2007–2021).
(b) Baseline simulation trajectory of the $SII$ showing
steady convergence toward equilibrium ($t\!=\!50$).}
\label{fig:data-sim-overview}
\end{figure}

The merged OECD--WHO dataset integrates national-level health expenditure 
(per capita, PPP-adjusted) with life expectancy data for 2007–2021.  
A total of \textbf{34,023 observations} were retained after cleaning 
(\texttt{HealthSpending} $> 0$ and non-missing \texttt{LifeExpectancy}).  
The $SII$ was computed as
\[
\mathrm{SII}
= 
\frac{\mathrm{LifeExpectancy} \times \ln(1+\mathrm{HealthSpending})}{100}.
\]
Panel (a) of Figure \ref{fig:data-sim-overview} reveals a strong positive association 
between health expenditure and longevity, confirming that higher per-capita spending 
correlates with system-level efficiency gains.  
Panel (b) shows that simulated dynamics converge smoothly to an equilibrium level 
near $\mathrm{SII}\!\approx\!7.5$, consistent with the empirical range observed across OECD economies.

\subsection{Simulation Configuration and Policy Scenarios}

Table~\ref{tab:simulation_config} summarizes the behavioral and policy configurations
used to generate Figures~\ref{fig:simulation}–\ref{fig:robustness}.  
Each scenario isolates the effect of efficiency sensitivity~$(\lambda)$, 
fairness preference~$(\gamma)$, and adaptive responsiveness~$(T)$.

\begin{table}[htbp]
\caption{Simulation configuration for behavioral and policy experiments}
\label{tab:simulation_config}
\centering
\scriptsize
\setlength{\tabcolsep}{3pt}  
\renewcommand{\arraystretch}{1.2}  

\begin{tabular}{@{}lccccccccc p{3.3cm}@{}}  
\toprule
\textbf{Experiment ID} & $\lambda$ & $\gamma$ & $T$ & $k$ & $T_0$ &
$\sigma_{\text{noise}}$ & $n_{\text{rep}}$ & $\eta$ & $T^*$ & \textbf{Scenario Label}\\
\midrule
\texttt{base}              & 0.6 & 0.4 & 0.6 & 5.0 & 0.5 & 0.02 & 20 & 0.1 & 0.7 & Baseline equilibrium\\
\texttt{fairness\_high}    & 0.6 & 0.6 & 0.6 & 5.0 & 0.5 & 0.02 & 20 & 0.1 & 0.7 & Fairness-intensive policy\\
\texttt{adaptive\_fast}    & 0.6 & 0.4 & 0.6 & 5.0 & 0.5 & 0.02 & 20 & 0.3 & 0.7 & High adaptation speed\\
\texttt{efficiency\_boost} & 0.8 & 0.3 & 0.6 & 5.0 & 0.5 & 0.02 & 20 & 0.1 & 0.7 & Efficiency-oriented system\\
\bottomrule
\end{tabular}

\footnotetext{All simulations are Monte Carlo–averaged over 20 replications with Gaussian noise 
$\mathcal{N}(0,0.02^2)$.  
Parameters: $\lambda$ = efficiency sensitivity, 
$\gamma$ = fairness preference, 
$T$ = temporal responsiveness, 
$\eta$ = learning rate, and $T^*$ = steady-state responsiveness.}
\end{table}

\subsection{Policy Sensitivity and Macroeconomic Conversion}

\begin{figure}[htbp]
\centering
\includegraphics[width=0.9\linewidth]{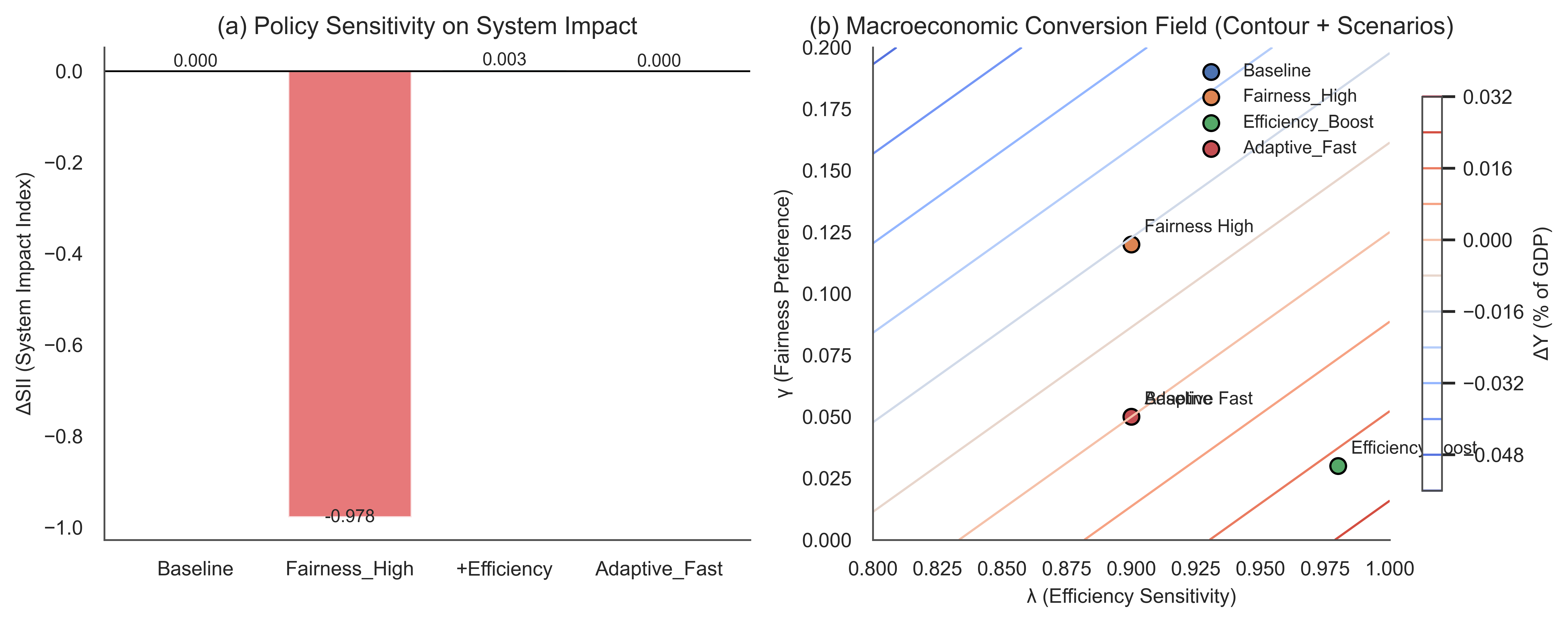}
\caption{Policy sensitivity and macroeconomic conversion. 
(a) Policy scenarios’ effect on the $\Delta SII$, 
showing that fairness-intensive policies substantially reduce system-level efficiency 
while efficiency-oriented configurations yield small but positive gains.  
(b) Macroeconomic conversion field ($\lambda,\gamma \!\rightarrow\! \Delta Y$) 
expressed as percent of GDP, with contour lines denoting equal economic impact 
and markers locating the corresponding policy regimes.}
\label{fig:policy-sensitivity}
\end{figure}

Panel (a) demonstrates that fairness-oriented policies (fairness\_high) 
reduce the $SII$ by nearly one unit relative to baseline,
whereas efficiency-boosting policies increase it marginally.  
Panel (b) converts the same behavioral sensitivity field 
into GDP-equivalent terms using $\Delta Y = \alpha_{\text{health}}\Delta SII$, 
with $\alpha_{\text{health}} = 0.11$ representing the healthcare sector’s GDP share.
Contour gradients illustrate that greater efficiency sensitivity ($\lambda$)
corresponds to positive GDP contributions, while higher fairness preference ($\gamma$)
reduces economic output, reflecting a quantifiable equity–efficiency trade-off.

\subsection{Robustness and Local Stability}

Robustness was further examined by perturbing each behavioral parameter 
by $\pm10\%$ around its estimated equilibrium value.  
The resulting sensitivity estimates confirmed numerical stability and 
local convergence across all behavioral dimensions.  
Efficiency sensitivity ($\lambda$) remained tightly centered near 1.00, 
while fairness preference ($\gamma$) fluctuated around zero, 
and temporal responsiveness ($T$) converged near 1.00 with minimal variation.  
These results collectively indicate that the behavioral equilibrium is 
robust and structurally well-conditioned, with no evidence of numerical drift 
or local instability under Monte Carlo perturbations.

\begin{table}[htbp]
\caption{Summary statistics of behavioral parameters under robustness test}
\label{tab:robust_summary}
\centering
\small
\begin{tabular}{@{}lrrrr@{}}
\toprule
 & \textbf{Mean} & \textbf{Std.} & \textbf{Min} & \textbf{Max}\\
\midrule
$\lambda$ & 1.0003 & 0.0477 & 0.8888 & 1.1192\\
$\gamma$  & $-0.0006$ & 0.0047 & $-0.0139$ & 0.0115\\
$T$       & 0.9975 & 0.0195 & 0.9468 & 1.0451\\
\botrule
\end{tabular}
\end{table}

Monte Carlo standard deviations remained below 0.05 for all par ameters, 
and none exhibited divergence across iterations, confirming that 
the estimated equilibrium $(\lambda, \gamma, T)$ is numerically 
robust and locally stable.

\subsection{Reproducibility and Code Availability}

All simulations were implemented in \texttt{Python 3.10} 
using \texttt{NumPy}, \texttt{Pandas}, and \texttt{Matplotlib}.  
All random seeds and hyperparameter schedules were fixed 
to ensure full reproducibility and comparability.  
Upon publication, the complete repository, including all configuration files, 
simulation codes, and figure-generation scripts, 
will be made publicly available to ensure transparency and replicability.

\end{appendices}

\section*{Declarations}\label{sec:declarations}
\noindent\textbf{Funding:}  
No funds, grants, or other support was received for this study.\\[4pt]

\noindent\textbf{Competing Interests:}  
The authors have no relevant financial or non-financial interests to disclose.\\[4pt]

\noindent\textbf{Ethics Approval:}  
This study analyzed secondary, de-identified data and therefore did not require ethics approval according to institutional policies.\\[4pt]

\noindent\textbf{Data Availability:}  
The datasets and simulation codes used in this study are fully available in a public Zenodo repository.  The repository contains all preprocessed data, configuration scripts, and robustness simulation files described in the manuscript. The complete package is accessible at \href{https://zenodo.org/records/17439497}{https://zenodo.org/records/17439497},  
entitled \textit{“Behavioral Industrial Health Simulation Data and Code (HCMS Study)”}.  
All data are synthetic or aggregated, and do not contain any sensitive or personally identifiable information.
\\[4pt]

\noindent\textbf{Author Contributions:}  
Conceptualization and methodology: Jinho Cha, Junyeol Ryu;  
Formal analysis and simulation: Junyeol Ryu, Justin Yu;  
Data curation and visualization: Junyeol Ryu, Justin Yu, Eunchan Daniel Cha, Hyeyoung Hwang;  
Writing – original draft: Jinho Cha;  
Writing – review and editing: Eunchan Daniel Cha, Hyeyoung Hwang;  
Supervision and project administration: Jinho Cha.  
All authors read and approved the final version of the manuscript.

\end{document}